\documentclass[letterpaper,10pt,conference]{ieeeconf}
\IEEEoverridecommandlockouts
\usepackage[T1]{fontenc}
\usepackage{microtype}
\usepackage{amsmath,amssymb,mathtools,bm}
\usepackage{graphicx,booktabs,array,float,balance}
\usepackage{cite,url}
\usepackage[colorlinks=true,linkcolor=blue!55!black,citecolor=blue!55!black,urlcolor=blue!55!black]{hyperref}
\usepackage{xcolor,colortbl}
\makeatletter
\def\endtable{\end@float}
\def\endfigure{\end@float}
\def\@begintheorem#1#2{\topsep\z@\normalfont\trivlist
  \item[\hskip\labelsep{\bfseries #1\ #2:}]}
\def\@opargbegintheorem#1#2#3{\topsep\z@\normalfont\trivlist
  \item[\hskip\labelsep{\bfseries #1\ #2\ (#3):}]}
\makeatother
\usepackage[caption=false,font=footnotesize]{subfig}
\usepackage[ruled,vlined,linesnumbered]{algorithm2e}
\usepackage[nameinlink]{cleveref}
\definecolor{oursmint}{HTML}{E2F1EA}
\definecolor{losssalmon}{HTML}{F9E8E4}
\definecolor{controlblue}{HTML}{0D47A1}
\definecolor{baselinered}{HTML}{BF360C}
\newtheorem{problem}{Problem}
\newtheorem{assumption}{Assumption}
\newtheorem{theorem}{Theorem}
\newtheorem{proposition}{Proposition}
\crefname{section}{Sec.}{Secs.}
\Crefname{section}{Sec.}{Secs.}
\crefname{subsection}{Sec.}{Secs.}
\Crefname{subsection}{Sec.}{Secs.}
\crefname{subsubsection}{Sec.}{Secs.}
\Crefname{subsubsection}{Sec.}{Secs.}
\crefname{figure}{Fig.}{Figs.}
\Crefname{figure}{Fig.}{Figs.}
\crefname{table}{Tab.}{Tabs.}
\Crefname{table}{Tab.}{Tabs.}
\crefname{algorithm}{Alg.}{Algs.}
\Crefname{algorithm}{Alg.}{Algs.}
\crefname{definition}{Def.}{Defs.}
\Crefname{definition}{Def.}{Defs.}
\crefname{lemma}{Lem.}{Lems.}
\Crefname{lemma}{Lem.}{Lems.}
\crefname{proposition}{Prop.}{Props.}
\Crefname{proposition}{Prop.}{Props.}
\crefname{theorem}{Thm.}{Thms.}
\Crefname{theorem}{Thm.}{Thms.}
\crefname{corollary}{Cor.}{Cors.}
\Crefname{corollary}{Cor.}{Cors.}
\crefname{assumption}{Assum.}{Assums.}
\Crefname{assumption}{Assum.}{Assums.}
\crefname{remark}{Rem.}{Rems.}
\Crefname{remark}{Rem.}{Rems.}
\crefname{problem}{Prob.}{Probs.}
\Crefname{problem}{Prob.}{Probs.}
\newcommand{\R}{\mathbb R}
\newcommand{\E}{\mathbb E}
\newcommand{\Prob}{\mathbb P}
\newcommand{\F}{\mathcal F}
\newcommand{\A}{\mathcal A}
\newcommand{\K}{\mathcal K}
\newcommand{\U}{\mathcal U}
\newcommand{\ind}{\mathbf 1}

\newcommand{\pos}[1]{[#1]_+}

\title{\LARGE\bf Belief-Informed Hybrid Control with Almost-Sure Target-Set Convergence}
\author{Clinton Enwerem\textsuperscript{1}, Saleh Kemal\textsuperscript{1}, John S. Baras\textsuperscript{1}, Calin Belta\textsuperscript{1}%
\thanks{\textsuperscript{1}The authors are with the Institute for Systems Research, University of Maryland, College Park, MD, USA. Emails:~\texttt{\char`\{enwerem, skemal, baras, calin\char`\}@umd.edu}.}
}
\begin{document}
\setlength{\abovedisplayskip}{3pt}
\setlength{\belowdisplayskip}{3pt}
\setlength{\abovedisplayshortskip}{0pt}
\setlength{\belowdisplayshortskip}{3pt}
\setlength{\jot}{2pt}
\maketitle
\flushbottom
\predisplaypenalty=0
\allowdisplaybreaks[1]
\begin{abstract}
Controlling a hybrid system to a target set under parameter uncertainty can require informative actions that temporarily drive the system state away from the specified set. We propose a belief-informed dual-control algorithm that combines belief-space receding-horizon selection with an expected-decrease constraint on a nonnegative target-set progress function. To accommodate exploratory deviations, a scalar controller state bounds the constraint's cumulative slack. Our algorithm's two-step lookahead selection uses predicted observations to update the parameter belief before evaluating the subsequent action's admissibility. Under distance-comparison bounds, correct conditional prediction, and recursive feasibility, we prove almost-sure target-set convergence at decision times and bound both the sum of expected progress-function values and the expected neighborhood-entry time. Upper confidence bounds extend these convergence guarantees to bounded model samples with summable error probabilities. In planar regulation with an unknown control direction, we verify recursive feasibility: the proposed two-step selection reduces the Euclidean state norm below 0.01 within 11 decisions for either sign, whereas a myopic one-step selection loses admissibility immediately after zero input. In a simulated bimanual assembly task, our algorithm's one-step implementation uses clearance feedback to complete the assembly under nominal friction, yielding a 96.4\% lower infinity-norm relative-position error than the reference-tracking baseline at the method's completion time. At lower friction, its posterior conditioning successfully detects an empty admissible set, whereas a fixed-prior alternative admits an action that violates the conditional decrease constraint.
Project page: \url{https://clintonenwerem.com/belief-hybrid-control/}.

\end{abstract}

\section{Introduction}\label{sec:introduction}
Consider the bimanual assembly task in \cref{fig:example}, where a robot with anthropomorphic end-effectors must acquire two free-body parts, align and insert one into the other, and withdraw its end-effectors without perturbing the assembly. Continuous robot and object motion combines with contact- and pose-triggered mode transitions, giving a hybrid control system~\cite{goebel2012,lindemann2021}. Unknown friction affects the response to control inputs, while object motion and contact forces inform its estimation~\cite{enwerem2026var}. A probing motion can reveal slip but also displace a part from its desired pose. Requiring immediate progress can therefore exclude actions that improve subsequent decisions. We seek a feedback-action selector that anticipates these physical and informational effects while ensuring convergence to a target set satisfying task requirements. This coupling of regulation and identification also arises in autonomous driving with uncertain surrounding-vehicle models~\cite{nair2022} and industrial process control~\cite{larsson2015}.
\begin{figure}[t]
\centering
\includegraphics[width=\linewidth]{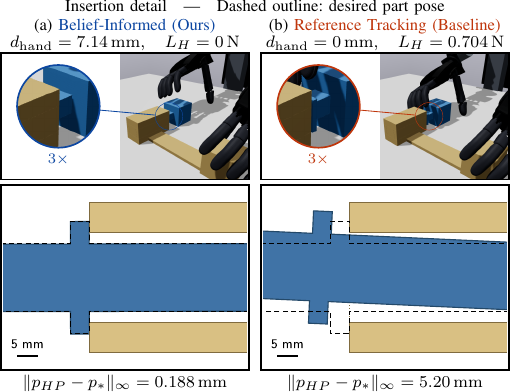}
\caption{\textbf{Part Assembly Outcomes.} We compare controllers on a contact-rich assembly task in which a bimanual robot with anthropomorphic end-effectors must align and insert a key-like part into a close-fitting housing (or socket) while accommodating uncertainty in the part--housing interaction. From the same initial state at \(\theta=0.7\), our belief-informed control algorithm completes the assembly with 96.4\% lower infinity-norm position error than a reference-tracking baseline. Upper panels show task endpoints with magnified insets. Lower panels show the key/socket sidewalls at a common scale together with dashed lines marking the desired part pose. Here \(p_{HP}\) and \(p_*\) denote the current and desired part-origin positions in the housing frame, \(d_{\rm hand}\) is the minimum hand--assembly clearance, and \(L_H\) is the summed magnitude of hand--housing contact forces.}
\label{fig:example}
\end{figure}

\subsection{Related Work}\label{sec:related-work}
In stochastic systems, control inputs exert a dual effect: they drive the system toward a target while reducing uncertainty that affects future decisions~\cite{barshalom1974}. Dual and adaptive predictive control methods address this coupling by anticipating future observations and parameter-estimation error~\cite{mesbah2018,klenske2016,heirung2017}, but at substantial computational cost. Tractable belief-space methods therefore introduce approximations such as nominal observations~\cite{platt2010}, Gaussian beliefs~\cite{vandenberg2012}, or localized sequential action optimization~\cite{nishimura2021} in place of full belief-state dynamic programming~\cite{smallwood1973,kaelbling1998}.

Uncertainty alone, however, does not capture systems that also undergo discrete mode transitions. Hybrid stability theory provides the corresponding framework for continuous flows, discrete jumps, and target sets~\cite{goebel2012}, encompassing switched systems~\cite{chatterjee2006}, rigid-body impacts~\cite{posa2016}, and stochastic recurrence~\cite{teel2013,teel2014}. Control Lyapunov functions and receding-horizon methods~\cite{primbs1999,mayne2000,scokaert1999} then provide tools for enforcing progress across such dynamics. Recent work extends these ideas to belief evolution in POMDPs~\cite{ahmadi2021}, bringing hybrid stability and stochastic estimation into a common setting.

Belief-space control functions regulate both task progress and estimation uncertainty, coupling information gathering with goal reaching and safety~\cite{vahs2026,enwerem2026}. We build on this perspective by defining a progress function over the observed hybrid state and using the parameter posterior to predict the endpoints of feedback actions through both flows and jumps. A conditional expected-decrease constraint allows temporary increases through bounded cumulative slack while preserving the convergence argument for recursively feasible selectors. Posterior-conditioned lookahead then ranks actions by possible next-observation outcomes, while hybrid barrier conditions~\cite{lindemann2021} restrict the candidate actions to those satisfying safety requirements.

\subsection{Contributions and Outline}\label{sec:contributions}
Our contributions are threefold. First, we formulate hybrid action selection with bounded cumulative slack and posterior-dependent two-step lookahead (\cref{sec:plant,sec:selection,sec:information}, \cref{alg:selector}). Second, under distance comparison, correct prediction, and recursive feasibility, we prove almost-sure target-set convergence and bound the sum of expected progress-function values and expected neighborhood-entry time (\cref{thm:exact}), with a sampled extension (\cref{prop:sample}). Third, we verify the convergence hypotheses and informative lookahead in planar regulation (\cref{sec:planar}), and evaluate one-step withdrawal feedback and posterior-dependent admissibility in bimanual assembly (\cref{sec:assembly}).

\section{Hybrid Plant and Belief Model}\label{sec:plant}
\subsection{Flows, Events, and Feedback Actions}\label{sec:hybrid-dynamics}
Following Goebel, Sanfelice, and Teel~\cite[Sec.~1.4]{goebel2012}, we use explicit modes, flow sets, guards, and resets. Let \(s=(x,q)\in S=X\times\mathcal Q\) denote the hybrid state, where \(x\in X\subseteq\R^n\) contains physical coordinates, controller memory, and clocks, and \(q\) is a mode in the finite set \(\mathcal Q\). The uncertain parameter \(\theta\in\Theta\subseteq\R^p\) is fixed throughout each closed-loop execution, and the mode-dependent control input satisfies \(u\in\R^{m_q}\). For mode \(q\), the plant flows according to
\begin{equation}\label{eq:flow}
\dot x=f_q(x,u,\theta),\quad x\in C_q(\theta),\quad u\in U_q(x).
\end{equation}
For a permitted edge \(e=(q,r)\), the guard \(D_e(\theta)\) and reset \(R_e\) specify the jump
\begin{equation}\label{eq:jump}
x^+=R_e(x,\nu,\theta),\quad q^+=r,\quad x\in D_e(\theta),
\end{equation}
where \(\nu\in\R^{p_e}\) is a reset input if the event is controlled. Commanded supervisory changes and forced contact events correspond to different edges. An event priority rule or probability law resolves simultaneous guards.

Let \(t_0=0\) and \(t_{k+1}=t_k+h_k\), where the action duration satisfies
\begin{equation}\label{eq:duration}
0<h_-\le h_k\le h_+<\infty \quad\text{almost surely}.
\end{equation}
At \(t_k\), we select a feedback action \(a_k\in\U(s_k)\) from a catalog of at most \(M\) elements. Each action specifies a feedback law that generates the control input \(u\) on \([t_k,t_{k+1})\), permitted mode switches, responses to forced events, and a stopping rule using interval observations. Fixed-period actions have \(h_k=h\); event-terminated actions stop at the first qualifying transition or time limit. We evaluate \(s_k\) after all events assigned to \(t_k\). For fixed \(\theta\), we model the response to \(a\) from \(s\) by the joint endpoint and observation kernel
\begin{equation}\label{eq:kernel}
P_\theta(ds',do\mid s,a).
\end{equation}
Here \(s'\) is the hybrid endpoint, including the mode and any clocks required by the model. The observation \(o\) contains the elapsed duration and interval observations such as contact outputs or event times. We assume measurable, complete trajectories with finitely many jumps in bounded physical time and use the plant's conditional transition law as \eqref{eq:kernel} in the convergence analysis.

\subsection{Parameter Belief and Observations}\label{sec:parameter-belief}
We infer \(\theta\) from the observed state, mode, and interval observations. Let \(\F_k\) contain observed states, selected feedback actions, and interval observations available before selecting \(a_k\), together with past computational random variables. With posterior \(b_k(d\theta)=\Prob(\theta\in d\theta\mid\F_k)\), we predict outcomes using
\begin{equation}\label{eq:predict}
P_b(ds',do\mid s,a)=\int_\Theta P_\theta(ds',do\mid s,a)b(d\theta).
\end{equation}
We assume \(P_\theta\) admits likelihood \(\ell_\theta(s',o\mid s,a)\), a probability mass or density relative to a common reference measure for each \((s,a)\). After observing \((s',o)\), we update \(b\) by Bayes' rule:
\begin{equation}\label{eq:bayes}
\tau(b,s,a,s',o)(d\theta)=
\frac{\ell_\theta(s',o\mid s,a)b(d\theta)}
{\int\ell_\vartheta(s',o\mid s,a)b(d\vartheta)}.
\end{equation}
We apply this update when its denominator is positive. By conditioning on the observed state as well as the interval observations, we include information about \(\theta\) contained in the physical response. Since \(\theta\) is fixed during an execution, we retain each sampled parameter value across prediction intervals. For numerical evaluation (\cref{sec:experiments}), we specialize to \(N_\theta\) deterministic parameter hypotheses \(\Theta=\{\theta_1,\ldots,\theta_{N_\theta}\}\). Under \(\theta_i\), we predict \(y=(s',o)=\Psi_i(s,a)\) and use the noiseless likelihood
\begin{equation}\label{eq:atomic-likelihood}
\ell_{\theta_i}(y\mid s,a)=\ind\{y=\Psi_i(s,a)\}.
\end{equation}
We retain the relative prior weights of matching hypotheses; a unique match receives unit posterior weight.

\section{Target-Set Control Problem}\label{sec:selection}
\subsection{Target Set and Progress Function}\label{sec:target}
Let \(\A\subseteq S\) be a target set in both continuous coordinates and mode. Given a metric \(d_S\), write \(d_\A(s)=\inf_{z\in\A}d_S(s,z)\). Coordinate scaling and rotation distance form part of the metric specification.
\begin{assumption}[Target Set Properties]\label{asm:anonemp}
The target set \(\mathcal{A}\) is nonempty and closed, and there exists a nonnegative measurable function \(V:S\to\R_{\ge0}\) satisfying
\begin{equation}\label{eq:comparison}
\alpha_1(d_\A(s))\le V(s)\le\alpha_2(d_\A(s))
\end{equation}
on the reachable domain, where \(\alpha_1,\alpha_2\) are continuous, strictly increasing, unbounded functions that vanish at zero (class \(\mathcal K_\infty\)). In particular, \(V(s)=0\) exactly when \(s\in\A\).
\end{assumption}
\begin{problem}[Target-Set Convergence]\label{prob:target-set}
Under \cref{asm:anonemp}, design a selector for which \(d_\A(s_k)\to0\) almost surely under the joint law of the prior, observations, and plant.
\end{problem}

\subsection{Admissible Actions and Expected Decrease}\label{sec:progress-allowance}
With decrease fraction \(c\in(0,1)\) and \(\rho=1-c\), we define the posterior expected endpoint value and required slack:
\begin{align}
\mu(s,b,a)&=\int V(s')P_b(ds',do\mid s,a),\label{eq:mu}\\
e(s,b,a)&=\pos{\mu(s,b,a)-\rho V(s)}.\label{eq:expense}
\end{align}
At decision \(k\), write \(\mu_k(a)=\mu(s_k,b_k,a)\) and \(e_k(a)=e(s_k,b_k,a)\). Here \([z]_+=\max\{z,0\}\), and \(e(s,b,a)\) is the nonnegative slack required when the predicted decrease is smaller than \(cV(s)\). An action with \(\mu(s,b,a)=V(s)\) requires slack \(cV(s)\), whereas an action with \(\mu(s,b,a)\le\rho V(s)\) requires none. Let \(B_0\ge0\) be a finite bound on cumulative slack, and let \(B_k=B_0-\sum_{j<k}\varepsilon_j\) denote the remaining slack at decision \(k\). With \(B\ge0\), the admissible action set is
\begin{equation}\label{eq:admissible}
\K(s,b,B)=\{a\in\U(s):e(s,b,a)\le B\}.
\end{equation}
To enforce the cumulative bound, we require \(\F_k\)-measurable action indices \(a_k\in\K(s_k,b_k,B_k)\) and set \(B_{k+1}=B_k-\varepsilon_k\), where \(\varepsilon_k=e(s_k,b_k,a_k)\). Admissibility and this slack update keep \(B_k\ge0\), giving \(\sum_k\varepsilon_k\le B_0\) on every trajectory. Under \eqref{eq:predict}, admissibility then implies
\begin{equation}\label{eq:drift}
\E[V(s_{k+1})\mid\F_k]\le\rho V(s_k)+\varepsilon_k.
\end{equation}
The decrease condition covers the flows and jumps over the complete action interval.

\begin{assumption}[Recursive Feasibility]\label{asm:feasible}
For every reachable decision history, including histories after observation and slack updates, the selector returns an admissible action with finite conditional expected endpoint value of \(V\). Its operation continues indefinitely on the modeled domain.
\end{assumption}
\Cref{asm:feasible} concerns the augmented state \((s,b,B)\). One sufficient construction restricts selection to finite-cost admissible actions whose successor augmented states belong almost surely to a controlled invariant set.

\subsection{Convergence Under Exact Prediction}\label{sec:exact-convergence}
The slack update offsets any permitted increase in the conditional expected value of \(V\). We apply the supermartingale argument underlying the Robbins--Siegmund theorem~\cite{robbins1971} to the nonnegative augmented function \(W_k=V(s_k)+B_k\).
\begin{theorem}[Almost-Sure Convergence]\label{thm:exact}
Assume well-defined evolution, the correct posterior kernel, \eqref{eq:comparison}, \cref{asm:feasible}, finite \(B_0\), and \(\E V(s_0)<\infty\). For every integer \(N\ge1\),
\begin{equation}\label{eq:bound}
c\sum_{k=0}^{N-1}\E V(s_k)\le\E V(s_0)+B_0.
\end{equation}
Consequently \(V(s_k)\to0\) almost surely and in mean and \(d_\A(s_k)\to0\) almost surely. For \(r>0\) and \(T_r=\inf\{k\ge0:d_\A(s_k)\le r\}\),
\begin{equation}\label{eq:entry}
\E T_r\le\frac{\E V(s_0)+B_0}{c\alpha_1(r)}.
\end{equation}
\end{theorem}
\begin{proof}
With \(W_k=V(s_k)+B_k\ge0\), \eqref{eq:drift} and the update of \(B_k\) give
\begin{equation}\label{eq:supermart}
\E[W_{k+1}\mid\F_k]\le W_k-cV(s_k).
\end{equation}
Taking expectations and summing over \(k=0,\ldots,N-1\) gives \eqref{eq:bound} because \(W_N\ge0\). Monotone convergence gives \(\E\sum_kV(s_k)\le(\E V(s_0)+B_0)/c\). This nonnegative series is finite almost surely, so \(V(s_k)\to0\) almost surely; summability of its expectations also implies convergence in mean. The lower comparison in \eqref{eq:comparison} gives convergence in distance. For each \(k<T_r\), \(V(s_k)\ge\alpha_1(r)\). Sum these inequalities and take expectations to obtain \eqref{eq:entry}.
\end{proof}

By \eqref{eq:duration}, the expected physical entry time also satisfies \(\E[t_{T_r}]\le h_+\E[T_r]\). Writing \(s(t,j)\) for the hybrid trajectory at physical time \(t\) after \(j\in\mathbb N_0\) jumps, if a continuous increasing function \(\chi\), with \(\chi(0)=0\), also bounds the maximum excursion within each interval by
\begin{equation}\label{eq:intersample}
\sup_{(t,j):t\in[t_k,t_{k+1})}d_\A(s(t,j))
\le\chi(V(s_k)+\varepsilon_k),
\end{equation}
then \(V(s_k)\to0\) and \(\varepsilon_k\to0\) imply convergence throughout physical time. 

\section{Belief-Informed Action Selection}\label{sec:information}
The expected-decrease constraint defines admissible actions but does not rank their information value. Actions requiring the same slack may yield observations supporting different subsequent feedback actions, so we condition the second decision on each predicted outcome. Let \(L(s,a,s')\ge0\) be a stage cost and \(\gamma\ge0\). For \(a\in\K(s,b,B)\), set \(B'=B-e(s,b,a)\). We define the second-stage cost-to-go \(H\) and two-step objective \(J_2\) by
\begin{align}
H(s',b',B')={}&\min_{a'\in\K(s',b',B')}\int\!\bigl[L(s',a',s'')\nonumber\\
&+\gamma V(s'')\bigr]P_{b'}(ds'',do'\mid s',a'),\label{eq:H}\\
J_2(s,b,B;a)={}&\int\!\bigl[L(s,a,s')\nonumber\\
&+H(s',\tau(b,s,a,s',o),B')\bigr]\nonumber\\
&\hspace{25mm}P_b(ds',do\mid s,a).\label{eq:J}
\end{align}
Because \(H\) uses the updated posterior, the second action's cost and admissibility depend on the first action's predicted observation, as in dual control~\cite{barshalom1974,klenske2016}. We assign \(H=+\infty\) when the second-stage admissible set is empty. We minimize \(J_2\) over admissible actions, resolve ties by catalog order, execute the first action, and replan after observing its outcome. \Cref{thm:exact} applies when the selector satisfies \cref{asm:feasible}. For comparison, the one-step selector minimizes \(\mu_k(a)\) over \(\K(s_k,b_k,B_k)\). Both examples use \(L=0\) and \(\gamma=1\).

\begin{algorithm}[t]
\caption{Belief-Space Receding-Horizon Action Selection}
\label{alg:selector}
\DontPrintSemicolon
\KwIn{\(s_k,b_k,B_k\); horizon \(H_p\in\{1,2\}\)}
\For{\(a\in\U(s_k)\)}{
  Compute \(\mu_k(a),e_k(a)\) using \eqref{eq:mu}--\eqref{eq:expense}\;
}
\(\K_k\leftarrow\{a\in\U(s_k):e_k(a)\le B_k\}\)\;
\lIf{\(\K_k=\varnothing\)}{\Return{no admissible action}}
\For{\(a\in\K_k\)}{
  \leIf{\(H_p=1\)}{\(J_k(a)\leftarrow\mu_k(a)\)}{\(J_k(a)\leftarrow J_2(s_k,b_k,B_k;a)\)}
}
\(\K_k^f\leftarrow\{a\in\K_k:J_k(a)<\infty\}\)\;
\lIf{\(\K_k^f=\varnothing\)}{\Return{no finite-cost admissible action}}
\(a_k\leftarrow\arg\min_{a\in\K_k^f}J_k(a)\)\;
\(B_{k+1}\leftarrow B_k-e_k(a_k)\)\;
Execute \(a_k\) until its stopping rule gives \(t_{k+1}\)\;
Observe \((s_{k+1},o_{k+1})\)\;
\(\Lambda_{k+1}\leftarrow\int\ell_\theta(s_{k+1},o_{k+1}\mid s_k,a_k)b_k(d\theta)\)\;
\lIf{\(\Lambda_{k+1}=0\)}{\Return{zero observation likelihood}}
\(b_{k+1}\leftarrow\tau(b_k,s_k,a_k,s_{k+1},o_{k+1})\)\;
\end{algorithm}

\subsection{Finite-Sample Action Selection}\label{sec:sampling}
When we estimate \eqref{eq:mu} from model samples, we use an upper confidence bound to avoid admitting actions through underestimated means. Assume \(0\le V(s')\le V_{\max}<\infty\). Conditional on \(\F_k\), draw \(Y_{a,i}=V(s'_{a,i})\) independently across \(i\) for each action from the model using an \(\F_k\)-measurable posterior \(\widehat b_k\); draws across actions may be dependent. The computational draws are conditionally independent of the true parameter and subsequent plant noise. Let \(d_{\rm TV}(\widehat b_k,b_k)\le\eta_k\), where \(d_{\rm TV}(b,\widehat b)=\sup_E|b(E)-\widehat b(E)|\) over measurable parameter sets, and let \(\eta_k\ge0\), \(\delta_k\in(0,1)\) be \(\F_k\)-measurable. With \(\overline Y_{a,n}=n^{-1}\sum_{i=1}^nY_{a,i}\), define
\begin{equation}\label{eq:ucb}
U_{k,a,n}=\overline Y_{a,n}+V_{\max}\!\left[\sqrt{\frac{\log(Mn(n+1)/\delta_k)}{2n}}+\eta_k\right].
\end{equation}
The posterior approximation changes the mean by at most \(V_{\max}\eta_k\). Hoeffding's inequality~\cite{hoeffding1963} and a union bound over actions and \(n\ge1\) give
\begin{equation}\label{eq:coverage}
\Prob\bigl(\mu_k(a)\le U_{k,a,n}\ \forall a,n\mid\F_k\bigr)\ge1-\delta_k,
\end{equation}
permitting data-dependent sample counts~\cite{howard2021}. We accept an action when \(\varepsilon_k=[U_{k,a,n}-\rho V(s_k)]_+\le B_k\) and set \(B_{k+1}=B_k-\varepsilon_k\). For two-step ranking, use this remaining slack in \eqref{eq:J} and exact or upper-bounded second-stage means conditional on each predicted state and posterior; recompute at the next executed decision.
\begin{proposition}[Summable Sampling Error]\label{prop:sample}
Under the physical and comparison conditions of \cref{thm:exact}, the sampling conditions above, and an accepted action returned after finitely many samples at every decision, suppose \(\sum_{k\ge0}\delta_k\le\Delta<\infty\) almost surely. Then
\begin{equation}\label{eq:samplebound}
c\sum_{k\ge0}\E V(s_k)\le\E V(s_0)+B_0+V_{\max}\Delta.
\end{equation}
The executed states satisfy \(d_\A(s_k)\to0\) almost surely, and \eqref{eq:entry} holds with \(V_{\max}\Delta\) added to its numerator.
\end{proposition}
\begin{proof}
Let \(\mathcal G_k\) extend \(\F_k\) by the full computational random stream. The selected action and sample count are \(\mathcal G_k\)-measurable, and conditional independence gives \(\E[V(s_{k+1})\mid\mathcal G_k]=\mu_k(a_k)\). On the coverage event in \eqref{eq:coverage}, the slack update gives \(\E[W_{k+1}\mid\mathcal G_k]\le W_k-cV(s_k)\); on its complement, \(W_{k+1}\le V_{\max}+B_k\). Thus
\begin{equation}\label{eq:sample-drift}
\E[W_{k+1}\mid\F_k]\le W_k-cV(s_k)+V_{\max}\delta_k.
\end{equation}
Summation and the argument of \cref{thm:exact} give the claims.
\end{proof}
For \(0<\Delta<1\), choosing \(\delta_k=\Delta/[(k+1)(k+2)]\) meets the summability condition.

\section{Numerical Experiments}\label{sec:experiments}
We use planar regulation to verify the hypotheses of \cref{thm:exact} and compare informative two-step selection with one-step selection. In bimanual assembly, we use one-step selection to examine withdrawal feedback and posterior-dependent action admissibility. Both examples evaluate the conditional means by summing over the finite deterministic hypotheses in \eqref{eq:atomic-likelihood}.
\subsection{Planar Regulation with Unknown Control Direction}\label{sec:planar}
Regulation with unknown control direction is a classical adaptive-control problem~\cite{nussbaum1983}. Consider the planar system
\begin{equation}\label{eq:planar-flow}
\dot x=\omega Jx+\theta u,\qquad
J=\begin{bmatrix}0&-1\\1&0\end{bmatrix},
\end{equation}
where \(x,u\in\R^2\), \(\omega=0.4\,\mathrm{rad/s}\), and the unknown control-direction sign \(\theta\in\{-1,+1\}\) is constant. With zero control input, the state rotates at constant radius. Radial feedback \(u=g_qx\) contracts or expands this radius according to the sign of \(\theta g_q\). Each feedback action selects one of four modes and applies its radial feedback law for one interval. The gains are
\begin{equation}\label{eq:planar-gains}
(g_{\mathrm{hold}},g_{\mathrm{test}},g_-,g_+)
=(0,-0.05,-0.5,0.5)\,\mathrm{s}^{-1}.
\end{equation}
At a commanded mode switch, \(x^+=x\). Let \(R(\varphi)=\exp(\varphi J)\) be the planar rotation matrix. For \(h=1\,\mathrm s\), the exact endpoint is \(x'=e^{\theta g_qh}R(\omega h)x\). The target is \(\A=\{(x,q):x=0\}\), and \(V(x,q)=\|x\|^2\), with \(\alpha_1(r)=\alpha_2(r)=r^2\). Since \(x^TJx=0\), the drift term preserves \(V\), and \(V'=e^{2\theta g_qh}V\). The noiseless hybrid endpoint \(s'=(x',q')\), together with the observed duration \(h\), supplies the likelihood in \eqref{eq:atomic-likelihood}. For \(x\ne0\), a nonzero gain identifies the control-direction sign through the observed change in radius. \Cref{fig:planar-scenario} shows the test interval and subsequent contracting trajectories.

\begin{figure}[t]
\centering
\includegraphics[width=.72\linewidth]{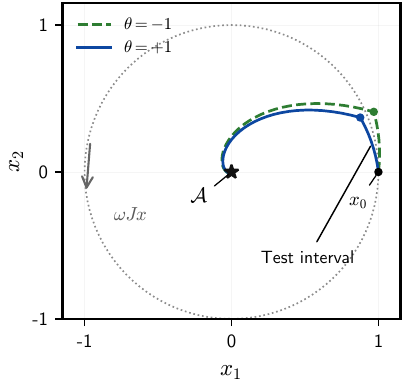}
\caption{\textbf{Planar Regulation with Unknown Control Direction.} In state coordinates \((x_1,x_2)\), \(\bigstar\) marks the target at the origin and the black dot marks \(x_0\). The dotted circle shows rotation under the drift \(\omega Jx\) with zero input. Two-step selection first applies the test action; the change in radius identifies the control-direction sign \(\theta\). We then select the contracting radial feedback gain. Blue solid and green dashed curves show the trajectories for \(\theta=+1\) and \(\theta=-1\), respectively; colored dots mark the test-interval endpoints.}
\label{fig:planar-scenario}
\end{figure}

We compare posterior-updating one- and two-step rules with a one-step fixed-prior rule, holding the admissibility constraint and catalog order fixed. Set \(x_0=(1,0)^T\), \(b_0(\pm1)=1/2\), \(c=0.1\), and \(B_0=0.12\). The test action gives \(\mu_{\rm test}=\cosh(0.1)\approx1.005\) and requires slack 0.105. High-gain actions require \(\cosh(1)-0.9\approx0.643>B_0\). The zero-input action minimizes the one-step mean but requires slack 0.1, leaving only 0.02 and no admissible next action.

With two-step selection, we first apply the test action. Its observation identifies \(\theta\), after which \(u=-0.5\theta x\) gives \(V(s_{k+1})=e^{-1}V(s_k)\). Its initial two-step objective is
\begin{equation}\label{eq:planar-score}
J_2(a_{\mathrm{test}})=e^{-1}\cosh(0.1)\approx0.37.
\end{equation}
The zero-input action has infinite second-stage cost because it leaves no admissible action at the next decision. After the test identifies the control-direction sign, the matching high-gain action requires zero slack at every subsequent decision. This verifies \cref{asm:feasible} along both possible closed-loop trajectories. The corresponding exact solution is
\begin{equation}\label{eq:planar-solution}
x_k=e^{-0.05\theta-0.5(k-1)}R(0.4k)x_0,\quad k\ge1.
\end{equation}
The exact solution satisfies \eqref{eq:bound}: the sum of expected progress-function values is \(\sum_{k\ge0}\E V(s_k)=1+\cosh(0.1)/(1-e^{-1})\approx2.59\), below the upper bound \((V(s_0)+B_0)/c=11.2\). Our simulated trajectories reproduce \eqref{eq:planar-solution} and first enter \(\|x\|\le0.01\) at decision 11 for both control directions.

For the prior and slack comparison, let \(p=b_0(-1)\) and \(\beta=B_0/\|x_0\|^2\). We vary \(p\in\{0.05,0.1,\ldots,0.95\}\) and \(\beta\in\{0,0.01,\ldots,0.2\}\), holding the plant, gains, interval length, initial state, and \(c\) fixed. We compare the first selected action for \(H_p=1\) and \(H_p=2\) (\cref{fig:planar}).
The one- and two-step selectors agree at 344 of 399 prior--slack pairs and differ at 55, all with \(0.5\le p\le0.7\). At 40 points, one-step selection chooses hold and two-step selection chooses test. At the remaining 15 points, hold is admissible but leaves no admissible next action, and the test action requires more than the available slack; consequently, no action has a finite two-step cost.

For example, at \(p=0.6\) and \(\beta=0.13\), the test requires slack 0.125. One-step selection chooses hold and leaves no admissible next action, whereas two-step selection chooses the test with cost 0.377. The observation identifies the contracting gain, which requires zero slack thereafter. Thus recursive feasibility also holds along both possible trajectories at this prior--slack pair.

\begin{figure}[t]
\centering
\begingroup
\setlength{\fboxrule}{0.8pt}\setlength{\fboxsep}{0pt}
\begin{minipage}[b]{.424\linewidth}
\centering
\subfloat[Embodiment\label{fig:scene-overview}]{%
\fbox{\includegraphics[trim=0bp 26bp 25bp 0bp,clip,width=\dimexpr\linewidth-2\fboxrule\relax]{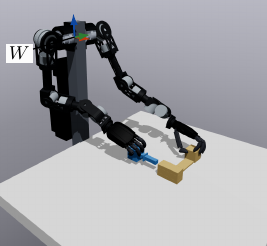}}}
\end{minipage}\hfill
\begin{minipage}[b]{.556\linewidth}
\centering
\subfloat[Assembly frames\label{fig:scene-detail}]{%
\fbox{\includegraphics[width=\dimexpr\linewidth-2\fboxrule\relax]{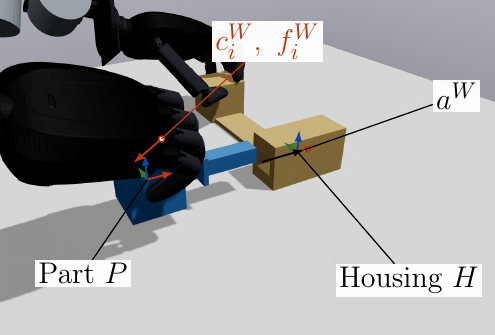}}}
\end{minipage}
\endgroup
\caption{\textbf{Assembly Task and Frames.} Panel (a) shows the bimanual RealHand A7 robot with L6 hands relative to the world frame \(W\). Panel (b) identifies the part frame \(P\), housing frame \(H\), and insertion axis \(a^W=R_H(1,0,0)^T\). At finger--part contact \(i\), the world-expressed contact position and force are \(c_i^W\) and \(f_i^W\). The part and housing poses are \(X_P=(R_P,p_P)\) and \(X_H=(R_H,p_H)\), respectively.}
\label{fig:scene}
\end{figure}
\begin{figure*}[t]
\centering
\subfloat[One-Step Selection, \(H_p=1\)]{%
\includegraphics[width=.39\textwidth]{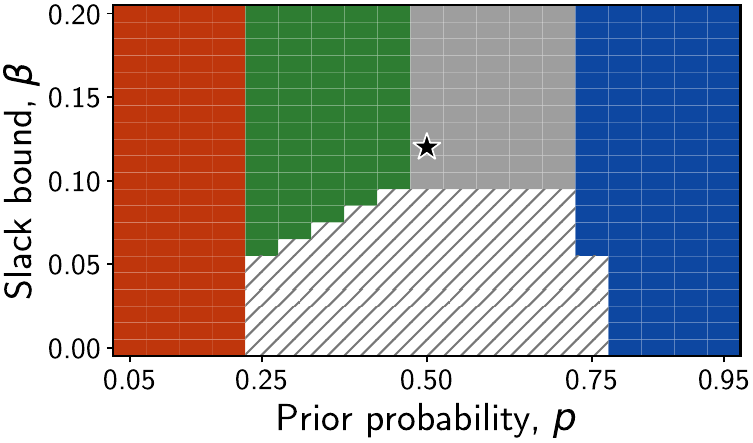}}\hfill
\subfloat[Two-Step Selection, \(H_p=2\)]{%
\includegraphics[width=.39\textwidth]{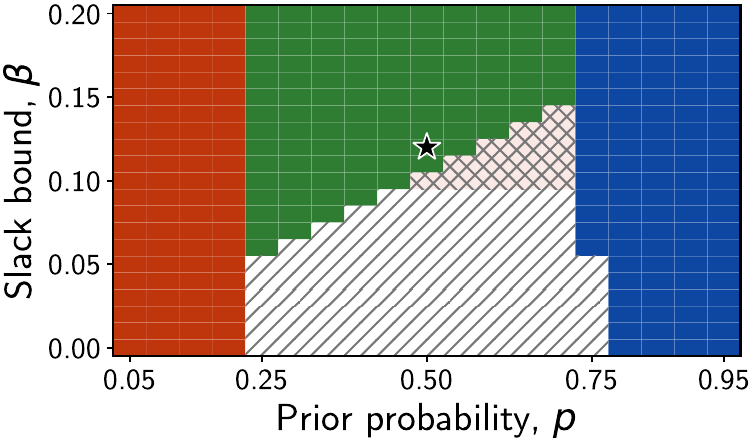}}\\[2pt]
\includegraphics[width=.60\textwidth]{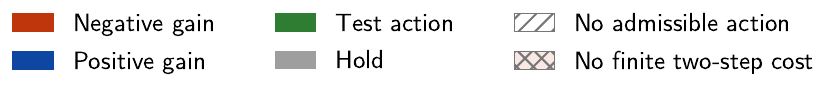}
\caption{\textbf{Action Selection Across Initial Beliefs and Slack Bounds.} The panels show the first actions selected over a 19-by-21 grid of prior probabilities \(p=b_0(-1)\) and normalized cumulative slack bounds \(\beta=B_0/\|x_0\|^2\). Colors identify the selected action. Diagonal hatching denotes an empty admissible set, while crosshatching denotes admissible first actions with no feasible second decision. Across panels, we hold the plant, gains, initial state, and expected-decrease constraint fixed. Stars mark the nominal case \((p,\beta)=(0.5,0.12)\).}
\label{fig:planar}
\end{figure*}
\subsection{Bimanual Assembly}\label{sec:assembly}
\subsubsection{State and Target}\label{sec:assembly-state}
We apply one-step selection to a bimanual robot comprising left and right variants of a RealHand A7 7-degrees-of-freedom (DoF) arm, each fitted with a matching RealHand L6 6-DoF robotic hand (\cref{fig:scene}). We simulate assembly in Drake~\cite{drake2019}, using hydroelastic contact~\cite{elandt2019}, the lagged approximation~\cite{castro2025}, and step \(h_{\rm sim}=1\,\mathrm{ms}\). We model the objects as free bodies under gravity and contact, with unknown friction \(\theta\) fixed during execution and shared across contact geometries. The observed continuous state contains robot configuration and velocity, object poses and twists, reference phase \(\sigma\in[0,1]\), elapsed time, guard/contact-loss clocks \(\tau_g,d_P,d_H\ge0\), dwell clock \(z\in[0,T_d]\), and contact outputs, with \(T_d=2\,\mathrm s\). The mode \(q=(q_c,\chi_H,\chi_P,\chi_L)\) combines \(q_c\in\{0,\ldots,9\}\) with binary housing-acquisition, part-acquisition, and lift flags. The flags start at zero and latch at their events; \(r(q)=9-q_c\).

\label{sec:assembly-target}
For object \(i\in\{P,H\}\), let \(X_i=(R_i,p_i)\in SE(3)\) be its world pose and \(v_i,\omega_i\in\R^3\) its world-expressed linear and angular velocities. Set \(p_{HP}=R_H^T(p_P-p_H)\), \(R_{HP}=R_H^TR_P\), \(p_*=(-0.0975,0,0)\,\mathrm m\), and \(e_R=\arccos((\operatorname{tr}R_{HP}-1)/2)\), corresponding to desired relative orientation \(I_3\). Let \(N_i,L_i,F_i\) denote hand-contact count, summed hand-contact-force magnitude, and upward table force. For object masses \(m_i\), define \(F_{\rm table}=F_P+F_H\), \(F_*=0.5(m_P+m_H)g_0\), with \(g_0=9.81\,\mathrm{m/s^2}\), and \(H_i=(N_i\ge2)\land(L_i\ge0.2\,\mathrm N)\). The minimum hand--assembly distance is \(d_{\rm hand}\). The completion predicate \(g(s)\) requires all residuals in \cref{tab:completion} to be nonpositive; \(g_{\rm ins}\) requires only its position and rotation conditions.
\begin{table}[t]
\centering
\caption{Completion Residuals and Normalization Scales. Values use SI units. For paired object rows, \(i=P,H\) in that order.}
\label{tab:completion}
\footnotesize
\setlength{\tabcolsep}{3pt}
\begin{tabular}{@{}lll@{}}\toprule
\(j\) & Residual \(c_j(s)\le0\) & Scale \(\lambda_j\)\\\midrule
1--3 & \(|p_{HP,j}-p_{*,j}|-\epsilon_{p,j}\) & \((0.1,0.05,0.05)_j\)\\
4 & \(e_R-0.06\) & 0.2\\
5, 6 & \(\|v_i\|-0.005\) & 0.1\\
7, 8 & \(\|\omega_i\|-0.05\) & 1\\
9 & \(0.005-d_{\rm hand}\) & 0.005\\
10 & \(F_*-F_{\rm table}\) & \(F_*\)\\\bottomrule
\multicolumn{3}{@{}l}{\(\epsilon_p=(0.003,0.003,0.002)\,\mathrm m\).}
\end{tabular}
\end{table}
The clock \(z\) is zero in modes 0--7. In dwell (\(q_c=8\)), each plant step sets \(z^+=\min\{T_d,z+h_{\rm sim}\}\) if the new state satisfies \(g\), and zero otherwise. At \(z^+=T_d\) with all three flags equal to 1, we enter terminal mode (\(q_c=9\)). We retain \(z=T_d\) there while \(g\) holds, and otherwise return to dwell with \(z^+=0\). The target is
\begin{equation}\label{eq:assemblytarget}
\A=\{s:g(s),\ z=T_d,\ r(q)=0,\ \chi_H=\chi_P=\chi_L=1\}.
\end{equation}

\subsubsection{Supervisory Modes}\label{sec:assembly-modes}
\Cref{tab:guards} defines forward transitions \(q_c\mapsto q_c+1\). Modes 1--5 require \(H_H\), and modes 2--5 require \(H_P\). Each contact-loss clock accumulates while its required predicate fails and resets otherwise; after 20~ms, we return to the corresponding acquisition mode, prioritizing housing recovery. We check recovery before updating \(\tau_g\), which accumulates while the forward guard and required contacts hold and resets otherwise. Forward transitions occur at 40~ms, and every mode transition resets \(\tau_g,d_P,d_H\). The transfer/alignment anchors are \(p_T=(-0.1675,0,0.06)\,\mathrm m\), \(p_A=(-0.1675,0,0)\,\mathrm m\), with component-wise tolerances \(\epsilon_T=(0.02,0.015,0.02)\,\mathrm m\), \(\epsilon_A=(0.02,0.008,0.008)\,\mathrm m\).

\begin{table}[t]
\centering
\caption{\scshape Forward Guards for Ordered Assembly Modes}
\label{tab:guards}
\footnotesize
\begin{tabular}{@{}>{\raggedright\arraybackslash}p{.06\linewidth}>{\raggedright\arraybackslash}p{.25\linewidth}>{\raggedright\arraybackslash}p{.58\linewidth}@{}}\toprule
\(q_c\) & Mode & Forward Predicate \\\midrule
0 & Housing acquisition & \(H_H\); housing speeds \(\le(0.05\,\mathrm{m/s},0.5\,\mathrm{rad/s})\)\\
1 & Part acquisition & \(H_P\land H_H\); part speeds \(\le(0.05\,\mathrm{m/s},0.5\,\mathrm{rad/s})\)\\
2 & Lift & \(p_{P,z}\ge-0.4441\,\mathrm m\), \(F_P\le0.12\,\mathrm N\)\\
3 & Transfer & \(|p_{HP}-p_T|\le\epsilon_T\), \(e_R\le0.12\,\mathrm{rad}\)\\
4 & Alignment & \(|p_{HP}-p_A|\le\epsilon_A\), \(e_R\le0.10\,\mathrm{rad}\)\\
5 & Insertion & \(g_{\rm ins}(s)\), \(F_{\rm table}\ge F_*\)\\
6 & Unloading & \(g_{\rm ins}(s)\), \(F_{\rm table}\ge F_*\), \(L_P,L_H\le0.01\,\mathrm N\)\\
7 & Withdrawal & Completion predicate \(g(s)\)\\\bottomrule
\end{tabular}
\end{table}

\begin{table}[t]
\centering
\caption{Controller and Simulation Parameters}
\label{tab:parameters}
\footnotesize
\begin{tabular}{@{}p{.57\linewidth}p{.35\linewidth}@{}}\toprule
Parameter & Value\\\midrule
Friction hypotheses; initial weights & \(\{0.35,0.7\}\); \((1/2,1/2)\)\\
\(c,B_0,H_p\) & \(10^{-3},0.1,1\)\\
Plant step, \(h_{\rm sim}\) & \(1\,\mathrm{ms}\)\\
Arm/hand stiffness (N\,m/rad) & 180 / 40\\
Arm/hand damping (N\,m\,s/rad) & 28 / 1.4\\

Action limits & \(0.1\,\mathrm s\), \(8\,\mathrm s\)\\
Guard/contact-loss times & \(40/20\,\mathrm{ms}\)\\
Terminal dwell, \(T_d\) & \(2\,\mathrm s\)\\
Insertion clearance per side & \(3\,\mathrm{mm}\)\\
Part/housing mass & \(0.12/0.35\,\mathrm{kg}\)\\
Point-contact stiffness & \(2\!\times\!10^5\,\mathrm{N/m}\)\\
Object hydroelastic modulus & \(10^7\,\mathrm{Pa}\)\\
\bottomrule
\end{tabular}
\end{table}
\subsubsection{Feedback Architecture}\label{sec:waypoints}
Contact trajectory optimization~\cite{posa2014} and whole-body predictive control~\cite{sleiman2021} provide motion-planning approaches. We use joint-limited inverse kinematics to generate 92 waypoints from prescribed palm poses and finger postures over a nominal reference duration \(T_{\rm ref}=33\,\mathrm s\), regularizing toward preceding configurations. We interpolate them with \(10\zeta^3-15\zeta^4+6\zeta^5\), where \(\zeta\in[0,1]\) is the normalized segment coordinate, to define \(q_a^{\rm ref}(\sigma)\). The independently commanded coordinates are \(q_a,v_a\in\R^{26}\), with joint limits \(q_a^-,q_a^+\); this subscript denotes actuated coordinates, whereas standalone \(a\) denotes a feedback action.

From lift through unloading, object-pose feedback corrects the right-arm reference toward \(X_E^*=X_P^{\rm ref}X_P^{-1}X_E\), where \(X_E\) is the current palm pose. We obtain \(X_P^{\rm ref}\) from the reference palm pose and the palm--part transform at acquisition, and compensate housing motion from transfer through unloading by premultiplying it by \(X_HX_{H,0}^{-1}\), with initial housing pose \(X_{H,0}\). We bound the regularized Jacobian correction by joint limits and \(0.2\,\mathrm{rad}\), then blend it smoothly with the waypoint reference. We generate the plant input \(u\in\R^{26}\) using joint impedance feedback~\cite{hogan1985} with an implicit-Euler velocity prediction, model torque limits, and the gains in \cref{tab:parameters}.

\subsubsection{Feedback Actions and Endpoint Prediction}\label{sec:assembly-actions}
\Cref{tab:actions} specifies the phase-rate multiplier \(v_\sigma\) and finger offset \(\delta_f\) for each feedback action. At each plant step, we advance normalized phase by \(v_\sigma h_{\rm sim}/T_{\rm ref}\), capped at 1 and the updated mode's endpoint. We preserve phase on forward transitions and restart acquisition on recovery. For tightening and relaxation, we add \(\delta_f\) to both hands' independently commanded pitch coordinates relative to the current waypoint command, then clip to joint limits; offsets do not accumulate. For clearance, once part--hand load is below 0.01~N, we translate the right-palm reference by \((-40,0,20)\,\mathrm{mm}\) in the housing frame and directly use the corrected right-arm command, retaining orientation, finger, and left-arm commands.

\begin{table}[t]
\centering
\caption{Feedback Action Catalog}
\label{tab:actions}
\footnotesize
\begin{tabular}{@{}lccc@{}}\toprule
Action & \(v_\sigma\) & \(\delta_f\) (rad) & Available \(q_c\)\\\midrule
Hold & 0 & 0 & 0--9\\
Advance & 1 & 0 & 0--7\\
Tighten & 0 & \(0.03\) & 1--5\\
Relax & 0 & \(-0.03\) & 1--5\\
Clearance & 1 & 0 & 6, 7\\
Clearance & 0 & 0 & 8, 9\\\bottomrule
\end{tabular}
\end{table}

In prediction and execution, we stop each action at the first supervisory transition or its duration limit. The limit is 8~s for advancement and clearance, 0.1~s for holding and finger adjustment, and 8~s for all available actions in dwell or terminal mode. Thus \(h_-=h_{\rm sim}\), \(h_+=8\,\mathrm s\). For each friction hypothesis, \(\Psi_i(s,a)\) composes plant/contact updates, feedback, reference and supervisory updates, and this stopping rule. The resulting state, clocks, phase, and contact outputs supply the likelihood in \eqref{eq:atomic-likelihood}.

\subsubsection{Progress Function}\label{sec:mode-rank}
We combine the mode rank \(r(q)\) with a nonnegative mode residual \(Z_q(x)\), specified below, to evaluate predicted endpoints:
\begin{equation}\label{eq:rank-progress}
V(x,q)=r(q)+\frac{Z_q(x)}{1+Z_q(x)},\qquad 0\le V<10.
\end{equation}
To measure deviation from the completion conditions, we define \(Z_\A\) from the normalized violations in \cref{tab:completion}, the dwell-time deficit, and unset acquisition/lift flags:
\begin{equation}\label{eq:terminal-residual}
Z_\A=\sum_{j=1}^{10}\left(\frac{[c_j(s)]_+}{\lambda_j}\right)^2
+\left(1-\frac z{T_d}\right)^2
+\sum_{i\in\{H,P,L\}}(1-\chi_i)^2.
\end{equation}
For each preceding mode, we obtain the endpoint joint reference \(q_a^q\) from the waypoint trajectory and define
\begin{equation}\label{eq:tracking-residual}
T_q=\frac1{26}\sum_{j=1}^{26}\left(\frac{q_{a,j}-q_{a,j}^q}{q_{a,j}^+-q_{a,j}^-}\right)^2
+\left(1-\min\left\{\frac{\tau_g}{0.04\,\mathrm s},1\right\}\right)^2.
\end{equation}
Using numerical values in SI units, set
\begin{align}
A_i={}&([2-N_i]_+/2)^2+([0.2-L_i]_+/0.2)^2\nonumber\\
 &+([\|v_i\|-0.05]_+/0.05)^2\nonumber\\
 &+([\|\omega_i\|-0.5]_+/0.5)^2,\label{eq:acquisition-residual}\\
Z_{\rm tr}(z_{\rm ref})={}&((p_{HP,x}+0.1675)/0.07)^2+(p_{HP,y}/0.05)^2\nonumber\\
 &+((p_{HP,z}-z_{\rm ref})/0.06)^2+(e_R/0.2)^2.\label{eq:transport-residual}
\end{align}
Here \(z_{\rm ref}\) is the desired part-origin height in the housing frame, expressed in metres. The mode residuals are
\begin{equation}\label{eq:mode-residuals}
Z_q=\begin{cases}
T_q+A_H,&q_c=0,\\
T_q+A_P,&q_c=1,\\
T_q+([-0.4441-p_{P,z}]_+/0.06)^2\\
\quad+([F_P-0.12]_+/1.2)^2,&q_c=2,\\
T_q+Z_{\rm tr}(0.06),&q_c=3,\\
T_q+Z_{\rm tr}(0),&q_c=4,\\
T_q+Z_\A,&q_c\in\{5,7\},\\
T_q+Z_\A+(L_P/0.2)^2+(L_H/0.2)^2,&q_c=6,\\
Z_\A,&q_c\in\{8,9\}.
\end{cases}
\end{equation}
Thus \(V^{-1}(0)=\A\). Every rank-reducing reset satisfies \(V(x^+,q^+)<r(q^+)+1\le r(q)\le V(x,q)\). We include recovery jumps in the action-endpoint prediction used in \eqref{eq:mu}.

\begin{figure}[t]
\centering
\begingroup
\setlength{\fboxrule}{0.8pt}\setlength{\fboxsep}{0pt}
\begin{minipage}[t]{.325\linewidth}\centering
\subfloat[Acquired]{\fbox{\includegraphics[trim=0bp 0bp 26bp 0bp,clip,width=\dimexpr\linewidth-2\fboxrule\relax]{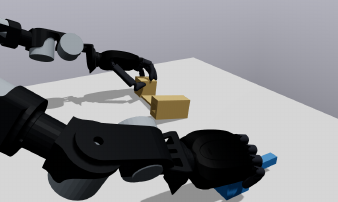}}}
\end{minipage}\hfill
\begin{minipage}[t]{.325\linewidth}\centering
\subfloat[Inserted]{\fbox{\includegraphics[trim=0bp 0bp 26bp 0bp,clip,width=\dimexpr\linewidth-2\fboxrule\relax]{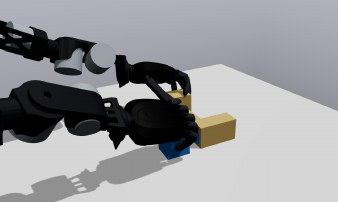}}}
\end{minipage}\hfill
\begin{minipage}[t]{.325\linewidth}\centering
\subfloat[Complete]{\fbox{\includegraphics[trim=0bp 0bp 26bp 0bp,clip,width=\dimexpr\linewidth-2\fboxrule\relax]{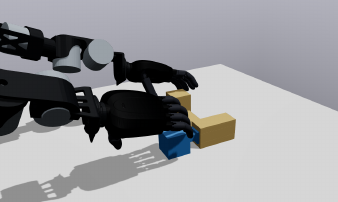}}}
\end{minipage}
\endgroup
\caption{\textbf{Assembly Sequence.} Both one-step selectors produce these states at friction \(\theta=0.7\). Physical finger contacts support acquisition and transport. Completion follows table support, hand withdrawal, and two seconds of stationary dwell.}
\label{fig:phases}
\end{figure}
\subsubsection{Comparison Protocol}\label{sec:results}
We compare posterior-updating and fixed-prior one-step selectors at \(\theta=0.35\) and 0.7 for at most 30 decisions. Across selectors, we hold the unacquired initial state, models, references, actions, guards, stopping rules, \(c\), and \(B_0\) fixed. At \(\theta=0.7\), we also evaluate a reference-tracking baseline at the selectors' completion time. For the baseline, we retain the plant, object-pose feedback, impedance law, and guards, but advance the reference within each mode and hold during dwell.

\subsubsection{Assembly Completion and Accuracy}
At \(\theta=0.7\), we obtain identical trajectories with both selectors: the robot maintains finger--object contact from lift through insertion and reaches \(\A\) (\cref{fig:phases}). Our selected actions require zero slack, and \(V\) decreases from 9.75 to zero. We identify friction from acquisition observations.

After insertion, the baseline continues reference advancement while our selectors choose clearance feedback. This feedback drives hand withdrawal while preserving the assembled part pose. At the selectors' completion time, we measure 96.4\% lower infinity-norm position error than the baseline and zero residual hand--housing load (\cref{fig:example,tab:assembly-results}).

\begin{table}[t]
\centering
\caption{Assembly Outcomes at \(\theta=0.7\). We evaluate all methods at 30.96~s, when the robot reaches \(\A\) under either selector with identical endpoints. Mint denotes selection; salmon denotes the reference-tracking baseline.}
\label{tab:assembly-results}
\footnotesize
\begin{tabular}{lrr}\toprule
Quantity & Selector & Baseline\\\midrule
Final progress, \(V\) & 0 & 3.94\\
Position error, \(\|p_{HP}-p_*\|_\infty\) (mm) & \cellcolor{oursmint}0.188 & \cellcolor{losssalmon}5.20\\
Rotation error (rad) & 0.00194 & 0.0453\\
Hand clearance (mm) & \cellcolor{oursmint}7.14 & \cellcolor{losssalmon}0\\
\textbf{Hand--housing load}, \(L_H\) (N) & \cellcolor{oursmint}\textbf{0} & \cellcolor{losssalmon}\textbf{0.704}\\
Total table support (N) & 4.61 & 4.30\\
Qualified terminal dwell (s) & \cellcolor{oursmint}2 & \cellcolor{losssalmon}0\\\bottomrule
\end{tabular}
\end{table}

\subsubsection{Belief Refinement and Action Admissibility}
At \(\theta=0.35\), we identify friction from acquisition observations, and both executions reach the same physical state in alignment mode. There, the minimum posterior slack requirement, 0.1143, exceeds the remaining 0.0936 (\cref{tab:admissibility}). Our posterior-updating selector therefore detects an empty admissible set. The fixed prior instead admits holding, tightening, and relaxation, and selects tightening. We obtain \(V'\approx5.7728\); because observations have identified friction and the transition is deterministic, this realized value equals the conditional mean. It exceeds \(\rho V+\varepsilon\) by 0.0543, violating \eqref{eq:drift}.

\begin{table}[t]
\centering
\caption{Action Evaluation at \(\theta=0.35\), \(V\approx5.6641\), and \(c=10^{-3}\). Remaining slack is 0.0936 (posterior) and 0.09318 (fixed prior). Mint and salmon distinguish their tightening evaluations.}
\label{tab:admissibility}
\footnotesize
\setlength{\tabcolsep}{3pt}
\begin{tabular}{lrrrr}\toprule
& \multicolumn{2}{c}{Endpoint mean, \(\mu\)} & \multicolumn{2}{c}{Required slack, \(e\)}\\
Action & Posterior & Fixed prior & Posterior & Fixed prior\\\midrule
Hold & 5.7805 & 5.7221 & 0.1220 & 0.0636\\
Advance & 8.9906 & 6.8842 & 3.3321 & 1.2258\\
Tighten & \cellcolor{oursmint}5.7728 & \cellcolor{losssalmon}5.7184 & \cellcolor{oursmint}0.1143 & \cellcolor{losssalmon}0.0600\\
Relax & 5.7900 & 5.7267 & 0.1315 & 0.0682\\\bottomrule
\end{tabular}
\end{table}

\subsection{Discussion}
In the planar-regulation example (\cref{sec:planar}), we use bounded slack for an informative action that enables subsequent contracting feedback. In the assembly example (\cref{sec:assembly}), our accuracy improvement follows from selecting clearance feedback after insertion, which commands hand withdrawal while preserving the assembled pose. We use the friction posterior to compute each action's expected endpoint value of \(V\). By conditioning on acquisition observations, we account for the slack that fixed-prior prediction underestimates at lower friction.

\section{Conclusion}\label{sec:conclusion}
We select hybrid feedback actions using posterior-dependent lookahead and an expected-decrease constraint with bounded cumulative slack. Under the stated prediction, comparison, and recursive-feasibility conditions, we prove almost-sure target-set convergence at decision times and bounds on summed expected progress-function values and expected neighborhood-entry time. Planar regulation verifies the hypotheses and demonstrates informative lookahead. In simulated assembly, one-step selection reduces endpoint position error through clearance feedback, while posterior conditioning detects loss of admissibility at lower friction.

\section*{Declaration of AI Use} We used generative AI tools for targeted technical code assistance, literature scoping, author-generated figure refinement, and selected manuscript text editing. The authors reviewed and verified the resulting material.

\bibliographystyle{ieeetr}
\bibliography{refs}
\end{document}